\documentclass[journal,twocolumn]{IEEEtranTCOM}
\usepackage{epsfig}
\usepackage{graphicx}
\usepackage{caption}
\usepackage{subcaption}
\usepackage{balance}
\usepackage{float}
\usepackage{psfrag}
\usepackage{amsfonts,amsmath,amssymb,mathrsfs,txfonts}

\newtheorem{lemma}{Lemma}

\usepackage{algorithm}
\usepackage{algorithmic}
\usepackage{url}
\usepackage{cite}
\usepackage{verbatim}
\usepackage{multirow}
\usepackage{multicol}
\usepackage[font=small]{caption}
\usepackage{pifont}
\usepackage{color}
\begin{document}
	\title{\huge A Successive Optimization Approach to Pilot Design for Multi-Cell Massive MIMO Systems}
	\author{Hayder Al-Salihi, Trinh Van Chien, Tuan Anh Le, and Mohammad Reza Nakhai
		\thanks{H. Al-Salihi and M. R. Nakhai are with the Department of Informatics, King's College London, London, U. K. T. V. Chien is with the Department of Electrical
		Engineering (ISY), Link\"{o}ping University, 581 83 Link\"{o}ping,
		Sweden. T. A. Le is with the Department of Design Engineering \& Mathematics, Middlesex University, London, U. K. (email: \{ana.al-salihi; reza.nakhai\}@kcl.ac.uk; trinh.van.chien@liu.se; t.le@mdx.ac.uk). The work of H. Al-Salihi was supported by Iraqi Higher Committee of Educational Development (HCED). The work of T. V. Chien was supported by the European Union’s Horizon 2020 research  and innovation programme under grant agreement No 641985 (5Gwireless), and also by ELLIIT and CENIIT.}
	}

	\maketitle

	\begin{abstract}
		In this letter, we introduce a novel pilot design approach that minimizes the total mean square errors of the minimum mean square error estimators of all base stations (BSs) subject to the transmit power constraints of individual users in the network, while tackling the pilot contamination in multi-cell Massive MIMO systems. First, we decompose the original non-convex problem into distributed optimization sub-problems at individual BSs, where each BS can optimize its own pilot signals given the knowledge of pilot signals from the remaining BSs. We then introduce a successive optimization approach to transform each optimization sub-problem into a linear matrix inequality (LMI) form, which is convex and can be solved by available optimization packages. Simulation results confirm the fast convergence of the proposed approach and prevails a benchmark scheme in terms of providing higher accuracy.	\end{abstract}
	\IEEEpeerreviewmaketitle
	\section{Introduction}
	In multi-cell Massive MIMO systems, each base station (BS) requires accurate knowledge of the channel state information (CSI) obtained during the pilot training phase. To attain accurate channel estimates, perfectly orthogonal pilot allocations to users are required. Unfortunately, this requirement is impractical, since the pilot overhead has to be proportional to the number of users in the entire system. Furthermore, the channel coherence block limits the number of orthogonal pilots \cite{Bjornson2016b}. Thus, pilot signals need to be reused over cells, causing spatially correlated interference, known as pilot contamination that degrades the performance of a Massive MIMO system \cite{Bjornson2016b}.

	In order to address the pilot contamination problem, the authors of \cite{Zhang2} proposed a superimposed channel estimation approach by adding a low power  pilot signal to the data signal at the transmitter. The superimposed signal is then utilized at the receiver for channel estimation. However, a proportion of the power allocated to the pilot signal is wasted. Fortunately, it has been shown in \cite{Ma} that the wasted-power problem can be theoretically mitigated with properly designed forward-error-correction codes. On the other attempts, pilot assignment and pilot power control are alternative solutions which can attain great improvements for the case that the system only has a finite set of orthogonal pilot signals \cite{PWang,Xu2015a, Chien2018a}. Involving reuse factor in pilot design may lead to a combinatorial pilot assignment problem in many pilot designs and, hence, result in an exponentially increased  computational complexity \cite{Mav1}.

In this letter, we consider a multi-cell Massive MIMO system adopting minimum mean square error (MMSE) estimators at BSs. We derive the mean square error (MSE) of the adopted MMSE estimator as a widely used accuracy criteria for estimation. We, then, formulate an optimization problem to find optimal pilot signals that minimize the total derived MSE of the MMSE estimators of all BSs in the network subject to a transmit power constraint at each user. The proposed formulation is non-convex with respect to the pilot matrices.
	To overcome non-convexity, we, first, decompose the proposed optimization problem into distributed subproblems at BSs, where each BS in the network optimizes its own pilot signal, given the knowledge of the pilot signals of other BSs. We then introduce a successive optimization approach to transform each subproblem into a linear matrix inequality (LMI) problem which is convex and can be effectively solved by available optimization packages, e.g., CVX \cite{Boyd}. Finally, we analyse the complexity of the transformed LMI optimization problem.
	
	\emph{\textbf{Notation}:}
	Bold lower/upper case letters are used for vectors/matrices; $\left\|\cdot\right\|_F$ and $\left\|\cdot\right\|$ stand for the Frobenius norm and the Euclidean norm; $(\cdot)^T$ and $(\cdot)^H$ is
	the regular and complex conjugate transpose operator, respectively; $\textrm{Tr}\left(\cdot
	\right)$ is the trace of a matrix;  $\mathbf{X}\succeq \mathbf{0}$ is the positive semidefinite condition;  $
	\mathbf{I}_a$ is an $a \times a$ identity matrix; $\textrm{diag}\{\mathbf{x}\}$ is a diagonal matrix which the diagonal entries are elements of the vector $\mathbf{x}$; $\mathcal{CN}(\cdot,\cdot)$ is a circularly symmetric
	complex Gaussian distribution; $\mathbb{E}[\cdot]$ is the expectation of a random variable;  $\mathcal{O}(\cdot)$ is the big-O notation.
	
	\section{System Model}
	Consider a multi-cell massive MIMO system with $C$ cells operating in a time-division duplexing  mode. Each cell comprises of an $M$-antenna BS and $N$ single-antenna users. The propagation factor between the $i$-th antenna of the BS in cell~$c^{*}$, denoted as BS~$c^{*}$, and user $n$ in cell $ c$ is $\sqrt{\phi^n_{c,c^*}} h^n_{c,c^*,i}$, where $\phi^n_{c,c^*}$ is the large scale fading coefficient modeling the path-loss and shadowing, while $h_{c,c^*,i}^n \sim \mathcal{CN}(0,1) $ is small-scale fading.
	
	In the pilot training phase, all users in each cell synchronously send their pilot signals.
	Let $\mathbf{x}^n_{c} \in \mathbb{C}^{\tau \times 1}$ be the pilot signal used by user $n$ in cell $c$ and $\| \mathbf{x}^n_{c} \|^2 \leq P_{\max,c}, \forall c$,	where $\tau$ is the length of the pilot signal, and $P_{\max,c}$ is the maximum allocated power level by each user in cell $c$ to its pilot signal. The received baseband training signal $\mathbf{y}_{c^*,i} \in \mathbb{C}^{\tau \times 1}$ at the $i$-th antenna element of the BS~$c^{*}$ can be expressed as:	
	\begin{eqnarray} \label{e1}
	\mathbf{y}_{c^*,i}&=&  \sum_{c=1}^{C} \sum_{n=1}^{N} \sqrt{  {{\phi}}^n_{c,c^*}}  {h}^n_{c,c^*,i} \mathbf{x}^n_{c} +\mathbf{v}_{c^*,i},
	\end{eqnarray}
	where $\mathbf{v}_{c^*,i}$ is Gaussian noise with $\mathbf{v}_{c^*,i}\sim \mathcal{CN}(\mathbf{0},\sigma^2\mathbf{I}_{\tau})$.
	Let the received signals, Gaussian noises, pilot signals by all antenna elements of  BS~$c^\ast$ and the corresponding large scale channel coefficients be denoted as
	
	\begin{align}
	&\mathbf{Y}_{c^*}  =[\mathbf{y}_{c^*,1},\ \mathbf{y}_{c^*,2}, \ldots, \mathbf{y}_{c^*,M}] \in \mathbb{C}^{\tau \times M},\label{e1v1} \\
	&\mathbf{V}_{c^*} =[\mathbf{v}_{c^*,1},\ \mathbf{v}_{c^*,2}, \ldots, \ \mathbf{v}_{c^*,M}] \in \mathbb{C}^{\tau \times M}, \label{e1v2} \\
	&\mathbf{X}_{c^*} =[\mathbf{x}^1_{c^*}, \ \mathbf{x}^2_{c^*}, \ldots,  \mathbf{x}^N_{c^*}]\in \mathbb{C}^{\tau \times N}, \label{e1v3}\\
	& \mathbf{D}_{c,c^*} =\textrm{diag}\{[{\phi}^1_{c,c^*}, \ \phi^2_{c,c^*}, \ldots, \phi^N_{c,c^*}]^T \}\in \mathbb{C}^{N \times N}. \label{e1v4}
	\end{align}
	Also, let the small-scale fading channel coefficients of all $N$ users in cell $c$ as seen by BS~$c^\ast$ be expressed as
	\begin{equation} \label{e1v5}
	\mathbf{H}_{c,c^*}=
	\begin{bmatrix}
	{h}^1_{c,c^*,1},& \ldots ,& {h}^1_{c,c^*,M} \\
	\vdots & \ddots &        ,        \vdots \\
	{h}^N_{c,c^*,1} & \ldots &, {h}^N_{c,c^*,M}
	\end{bmatrix} \in \mathbb{C}^{N \times M}.
	\end{equation}
	Then, using \eqref{e1} $-$ \eqref{e1v5}, one  can formulate the received training signals by all $M$ antenna elements of  BS $c^*$, according
	\begin{equation} \label{e2}
	\mathbf{Y}_{c^*} =  \mathbf{X}_{c^*}  {\mathbf{D}}^{\frac{1}{2}}_{c^*,c^*} \mathbf{H}_{c^*,c^*}  +  \sum_{c=1, c\neq{c^*}}^{C}  \mathbf{X}_{c}{\mathbf{D}}^{\frac{1}{2}}_{c,c^*}  \mathbf{H}_{c,c^*}  +\mathbf{V}_{c^*}.
	\end{equation}
	The first term in right hand side of \eqref{e2} involves desired channel coefficients and the remaining terms indicate the effects of mutual interference and Gaussian noise. The channel estimate $\widehat{\mathbf{H}}_{c^*,c^*}$ of the original channel $\mathbf{H}_{c^*,c^*}$ is computed by utilizing the MMSE estimation upon the observation of $\mathbf{Y}_{c^*}$ is \cite{Kailath}:
	\begin{equation} \label{e3}
	\widehat{\mathbf{H}}_{c^*,c^*} =  \mathbb{E} [\mathbf{H}_{c^*,c^*} \mathbf{Y}_{c^*}^H] \left( \mathbb{E}[\mathbf{Y}_{c^*}\mathbf{Y}_{c^*}^H]\right)^{-1} \mathbf{Y}_{c^*}.
	\end{equation}
	Plugging \eqref{e2} in \eqref{e3}, and after some mathematical manipulations, we obtain
	\begin{eqnarray}\label{e5}
	\widehat{\mathbf{H}}_{c^*,c^*} = M {\mathbf{D}}^{\frac{1}{2}}_{c^*,c^*} \mathbf{X}_{c^*}^H \mathbf{\Omega}_{c^\ast}^{-1}  \mathbf{Y}_{c^*},
	\end{eqnarray}
	where
	$\boldsymbol{\Omega}_{c^*}= M \sum_{c=1}^{C} \mathbf{X}_{c} \mathbf{D}_{c,c^*}\mathbf{X}_{c}^H + M \sigma^2 \mathbf{I}_\tau.$ From \eqref{e5}, the channel estimation quality depends on the pilot design and if $\tau < CN$, it also suffers from pilot contamination \cite{PWang,Xu2015a, Chien2018a}. Let  the channel estimation errors at BS~$c^\ast$ be denoted as
	\begin{equation}
	\mathbf{\Delta}_{c^\ast}={\mathbf{H}}_{c^*,c^*}- \widehat{\mathbf{H}}_{c^*,c^*},
	\end{equation}
	and the MSE be defined as
	\begin{eqnarray}
	\label{e6}
	\textrm{MSE}_{c^\ast} = \mathbb{E} \left[\|\mathbf{\Delta}_{c^\ast}\|^2_F\right]
	= \mathbb{E} \left[\textrm{Tr}\left(\mathbf{\Delta}_{c^\ast} \mathbf{\Delta}_{c^\ast}^H\right) \right].
	\end{eqnarray}
	Then, using \eqref{e2}, \eqref{e3}, \eqref{e5},  and after some mathematical manipulations, one can rewrite $\textrm{MSE}$ in \eqref{e6} as:
	\begin{eqnarray} \label{e7}
	M\textrm{Tr}\left(\mathbf{A}^{-1}-\mathbf{A}^{-1} \mathbf{B} (\mathbf{C}^{-1}+\mathbf{D A}^{-1}\mathbf{B})^{-1}\mathbf{D A}^{-1}
	\right),
	\end{eqnarray}
	where $\mathbf{A}^{-1}=\mathbf{I}_N$, $\mathbf{B}=M \mathbf{D}^{\frac{1}{2}}_{c^*,c^*} \mathbf{X}_{c^*}^H$, $\mathbf{D}=\mathbf{X}_{c^*} \mathbf{D}^{\frac{1}{2}}_{c^*,c^*}$, and $\mathbf{C}^{-1}=M \sum_{c=1, c\neq{c^*}}^{C} \mathbf{X}_{c} \mathbf{D}_{c,c^*}\mathbf{X}_{c}^H+M \sigma^2 \mathbf{I}_{\tau}$. By utilizing the Sherman-Morrison-Woodbury identity
	\[	(\mathbf{A}+\mathbf{BCD})^{-1}=\mathbf{A}^{-1}-\mathbf{A}^{-1} \mathbf{B} (\mathbf{C}^{-1}+\mathbf{D A}^{-1}\mathbf{B})^{-1}\mathbf{D A}^{-1}
	\]
	and defining $\textrm{MSE}_{c^\ast} \triangleq f_{c^*}\left( \mathbf{X}_{c^*}\right)$,  one can reformulate \eqref{e6} as 
	\begin{eqnarray} \label{e9}
	f_{c^*}\left( \mathbf{X}_{c^*}\right)=M\textrm{Tr}\left(\left(\mathbf{I}_N+ \mathbf{D}^{\frac{1}{2}}_{c^*,c^*}\mathbf{X}_{c^*}^H \mathbf{F}^{-1}_{c^*} \mathbf{X}_{c^*} \mathbf{D}^{\frac{1}{2}}_{c^*,c^*} \right)^{-1}\right),
	\end{eqnarray}
	where $\mathbf{F}_{c^*}=\sum_{c=1, c\neq{c^*}}^{C} \mathbf{X}_{c} \mathbf{D}_{c,c^*}\mathbf{X}_{c}^H+ \sigma^2 \textbf{I}_\tau$.
	
	\section{A Successive Optimization Pilot Design}
	From \eqref{e9}, the performance of the MMSE estimation algorithm depends on the pilot structure. In this section, we develop an optimal pilot design to minimize the total channel estimation errors of all BSs in the network subject to the transmit power constraints at individual users. Hence, we introduce the following optimization problem for the network:
	\begin{equation}
	\begin{aligned}\label{prob_01}
	& \underset{ \{\mathbf{X}_{c^*} \} }{\textrm{minimize}} & &
	\sum_{c^\ast =1}^C  f_{c^*}\left( \mathbf{X}_{c^*}\right)\\
	& \mbox{subject to}\ & & \mathbf{X}_{c^*}^H \mathbf{X}_{c^*} \preceq  P_{\max, c^\ast} \mathbf{I}_N, \forall c^{\ast},
	\end{aligned}
	\end{equation}
	where  $\{\mathbf{X}_{c^{\ast}}\}=\{\mathbf{X}_1, \mathbf{X}_2,\cdots, \mathbf{X}_C\}$.
	Problem \eqref{prob_01} is non-convex due to its objective function. To tackle the problem, we first introduce an auxiliary variable $\mathbf{G}_{c^*}$, denote $\{\mathbf{G}_{c^{\ast}}\}=\{\mathbf{G}_1,\cdots, \mathbf{G}_C\}$, remove the constant $M$, and rewrite \eqref{prob_01} as\footnote{Problems \eqref{prob_01} and \eqref{prob_02} are equivalent since \eqref{prob_02} is an epigraph form of \eqref{prob_01} \cite[pp.134]{Boyd_convex}. In fact, introducing the auxiliary variable $\mathbf{G}_{c^*}$ transforms the objective function into a linear form while shifting the nonlinear part into a constraint.}
	\begin{equation}\label{prob_02}
	\begin{aligned}
	&&& \displaystyle \underset{ \{\mathbf{X}_{c^*}\},\{\mathbf{G}_{c^*} \}}{\textrm{minimize}} \quad
	\sum_{c^{\ast} =1}^C \textrm{Tr}\left( \mathbf{G}_{c^*}\right)\\
	&&& \mbox{subject to} \quad \mathbf{X}_{c^*}^H \mathbf{I}_{\tau}^{-1} \mathbf{X}_{c^*} \preceq  P_{\max, c^\ast}\mathbf{I}_N, \forall c^{\ast},\\
	&&&\left(\mathbf{I}_N+  \mathbf{D}^{\frac{1}{2}}_{c^*,c^*}\mathbf{X}_{c^*}^H \mathbf{F}^{-1}_{c^*} \mathbf{X}_{c^*} \mathbf{D}^{\frac{1}{2}}_{c^*,c^*}\right)^{-1} \preceq \mathbf{G}_{c^*}, \forall c^\ast.
	\end{aligned}
	\end{equation}
	Adopting the Schur complement \cite{Boyd_convex}, one can  equivalently reformulate problem \eqref{prob_02} as
	\begin{equation}
	\begin{aligned}\label{prob_04}
	&&& \displaystyle \underset{ \{\mathbf{X}_{c^*}\},\{\mathbf{G}_{c^*} \}}{\textrm{minimize}} \quad
	\sum_{c^\ast =1}^C \textrm{Tr}\left( \mathbf{G}_{c^*}\right)\\
	&&& \mbox{subject to} \quad \begin{bmatrix}
	P_{\max, c^\ast} \mathbf{I}_N &\mathbf{X}_{c^*}^H\\
	\mathbf{X}_{c^*}& \mathbf{I}_{\tau}
	\end{bmatrix}\succeq \mathbf{0}, \forall c^\ast,\\
	&&& \begin{bmatrix} \mathbf{G}_{c^*} & \mathbf{I}_N\\ \mathbf{I}_N & \mathbf{I}_N+  \mathbf{D}^{\frac{1}{2}}_{c^*,c^*}\mathbf{X}_{c^*}^H \mathbf{F}^{-1}_{c^*} \mathbf{X}_{c^*} \mathbf{D}^{\frac{1}{2}}_{c^*,c^*}
	\end{bmatrix} \succeq \mathbf{0}, \forall c^\ast.
	\end{aligned}
	\end{equation}
	
	The second set of constraints in \eqref{prob_04} is still non-convex due to the nonlinearity of the term $\mathbf{D}^{\frac{1}{2}}_{c^*,c^*}\mathbf{X}_{c^*}^H \mathbf{F}^{-1}_{c^*} \mathbf{X}_{c^*} \mathbf{D}^{\frac{1}{2}}_{c^*,c^*}$ with respect to optimization variable $\mathbf{X}_{c^*}$, $\forall c^*$, i.e., the optimization variable is in quadratic forms and appears in both numerator and denominator of the term. As a main contribution of this paper, we propose a distributed algorithm where every BS $c^\ast$ optimizes its own pilot signals given the knowledge of the pilot signals of the other cells in $\mathbf{F}^{-1}_{c^*}$ as follows:
	\begin{equation}
	\begin{aligned}\label{prob_05v1}
	&&& \displaystyle \underset{ \mathbf{X}_{c^*},
		\mathbf{G}_{c^*}}{\textrm{minimize}} \quad
	\textrm{Tr}\left(\mathbf{G}_{c^*} \right)\\
	&&& \text{subject to} \quad \begin{bmatrix}
	P_{\max, c^\ast} \mathbf{I}_N & \mathbf{X}_{c^*}^H \\
	\mathbf{X}_{c^*} & \mathbf{I}_{\tau}
	\end{bmatrix}\succeq \mathbf{0},\\
	&&& \begin{bmatrix} \mathbf{G}_{c^*} & \mathbf{I}_N\\ \mathbf{I}_N & \mathbf{I}_N+  \mathbf{D}^{\frac{1}{2}}_{c^*,c^*} \mathbf{X}_{c^*}^{H} \mathbf{F}^{-1}_{c^*}\mathbf{X}_{c^*} \mathbf{D}^{\frac{1}{2}}_{c^*,c^*}
	\end{bmatrix} \succeq \mathbf{0}.
	\end{aligned}
	\end{equation}
	
	Although the distributed optimization problem \eqref{prob_05v1} only considers $\mathbf{X}_{c^*}$ and	$\mathbf{G}_{c^*}$ as the optimization variables, its second constraint is still not in an LMI form with respect to $\mathbf{X}_{c^*}$. To proceed, we propose a successive optimization approach where, at the $t$-th iteration, BS $c^\ast$ updates its pilot signals by solving the following distributed optimization problem:
	\begin{equation}
	\begin{aligned}\label{prob_05}
	&&& \displaystyle \underset{ \mathbf{X}_{c^*}^{(t)},
		\mathbf{G}_{c^*}^{(t)}}{\textrm{minimize}} \quad
	\textrm{Tr}\left(\mathbf{G}_{c^*}^{(t)}\right)\\
	&&& \text{subject to} \quad \begin{bmatrix}
	P_{\max,c^\ast} \mathbf{I}_N & \mathbf{X}_{c^*}^{(t),H} \\
	\mathbf{X}_{c^*}^{(t)} & \mathbf{I}_{\tau}
	\end{bmatrix}\succeq \mathbf{0},\\
	&&& \begin{bmatrix} \mathbf{G}_{c^*}^{(t)} & \mathbf{I}_N\\ \mathbf{I}_N & \mathbf{I}_N+  \mathbf{D}^{\frac{1}{2}}_{c^*,c^*} \mathbf{X}_{c^*}^{(t),H} (\mathbf{F}^{-1}_{c^*})^{(t-1)} \mathbf{X}_{c^*}^{(t-1)} \mathbf{D}^{\frac{1}{2}}_{c^*,c^*}
	\end{bmatrix} \succeq \mathbf{0},
	\end{aligned}
	\end{equation}
	where $\mathbf{F}_{c^*}^{-1}$ from the previous iteration is
	\begin{equation}\label{eqtF}
	(\mathbf{F}_{c^*}^{-1})^{(t-1)}=\sum_{c=1, c\neq{c^*}}^{C} \mathbf{X}_{c}^{(t-1)} \mathbf{D}_{c,c^*} \mathbf{X}_{c}^{(t-1),H}+ \sigma^2\textbf{I}_\tau,
	\end{equation}
	$\mathbf{X}_{c^*}^{(t-1)}$ and $\mathbf{X}_{c}^{(t-1)}$ are the optimal pilots of cells $c^*$ and $c$, respectively, which are obtained from the $(t-1)$-th iteration.
	In order to transform the second constraint of \eqref{prob_05} into an LMI form with respect to both $\mathbf{X}_{c^{\ast}}^{(t)}$  and $\mathbf{G}_{c^*}^{(t)}$, we have used the known value of  $\mathbf{X}_{c^{\ast}}^{(t-1)}$. Notice that at  the stationary point attained after a sufficient number of iterations, the approximation
	\begin{equation} \label{eq:Assumption}
	\mathbf{X}_{c^{\ast}}^{(t)} \approx  \mathbf{X}_{c^{\ast}}^{(t-1)}, \forall c^{\ast},
	\end{equation}
	can be assured with any desired accuracy.	Note that, the matrix in the second constraint of \eqref{prob_05} is not Hermitian during the iterations, due to the mismatch between $\mathbf{X}_{c^{\ast}}^{(t)}$  and $\mathbf{X}_{c^{\ast}}^{(t-1)}, \forall c^{\ast}$. To guarantee a Hermitian matrix in the second constraint of \eqref{prob_05}, we introduce a new variable $\mathbf{A}_{c^\ast}^{(t)}$, such that
	\begin{equation}\label{newconstraint}
	\begin{split}
	2\mathbf{A}_{c^\ast}^{(t)} =&  \mathbf{D}^{\frac{1}{2}}_{c^*,c^*} \mathbf{X}_{c^*}^{(t),H} (\mathbf{F}^{-1}_{c^*})^{(t-1)} \mathbf{X}_{c^*}^{(t-1)} \mathbf{D}^{\frac{1}{2}}_{c^*,c^*}\\
	&+ \mathbf{D}^{\frac{1}{2}}_{c^*,c^*} \mathbf{X}_{c^*}^{(t-1),H} (\mathbf{F}^{-1}_{c^*})^{(t-1)} \mathbf{X}_{c^*}^{(t)} \mathbf{D}^{\frac{1}{2}}_{c^*,c^*}.
	\end{split}
	\end{equation}
	Finally, we reformulate \eqref{prob_05} as
	\begin{equation}
	\begin{aligned}\label{prob_06}
	& \displaystyle \underset{ \mathbf{X}_{c^*}^{(t)},
		\mathbf{G}_{c^*}^{(t)},\mathbf{A}_{c^*}^{(t)}}{\textrm{minimize}} &&
	\textrm{Tr}\left(\mathbf{G}_{c^*}^{(t)}\right)\\
	& \text{subject to} && \begin{bmatrix}
	P_{\max,c^\ast} \mathbf{I}_N & \mathbf{X}_{c^*}^{(t),H} \\
	\mathbf{X}_{c^*}^{(t)} & \mathbf{I}_{\tau}
	\end{bmatrix}\succeq \mathbf{0},\\
	&&& \begin{bmatrix} \mathbf{G}_{c^*}^{(t)} & \mathbf{I}_N\\ \mathbf{I}_N & \mathbf{I}_N+  \mathbf{A}_{c^\ast}^{(t)}
	\end{bmatrix} \succeq \mathbf{0},\\
	&&& \textrm{constraint} \ \ \eqref{newconstraint}.
	\end{aligned}
	\end{equation}
	Problem \eqref{prob_06} is now convex and can be efficiently solved by CVX \cite{Boyd}. The procedure to obtain the optimal pilot signals for all $C$ cells in the network is summarized in Algorithm~\ref{alg2}.
{\it Remark 1 (Convergence):} Since problem~\eqref{prob_06} is convex, steps 3 and 4 in Algorithm~\ref{alg2} ensure the $\delta$-convergence of $\mathbf{X}_{c^*}^{(t)}$ to its optimal value and a minimal objective function value in problem~\eqref{prob_06} per cell.\footnote{ Although the global optimality can be achieved per iteration and per cell by solving \eqref{prob_06}, it may not be achievable to the original multicell problem~\eqref{prob_01} due to its inherent non-convexity. In fact, Algorithm~\ref{alg2} is a suboptimal algorithm with an affordable complexity to the NP-hard problem~\eqref{prob_01}.}
	\begin{algorithm}[t]
		\caption{Successive optimization approach for \eqref{prob_02}}
		\begin{algorithmic}[1]\label{alg2}
			\STATE {\bf Inputs:} $\mathbf{D}_{c,c^*}$, $P_{\max, c^\ast}$, $\sigma^2$, stopping criteria $\delta>0$, initialize $\mathbf{X}_{c}^{(0)}$, $\forall c, c^{\ast}$; $t=1$;
			\STATE Each cell $c^*$ calculates $\mathbf{F}_{c^*}^{(t-1)}$ utilizing \eqref{eqtF} and then solves \eqref{prob_06} to attain $\mathbf{X}_{c^*}^{(t)}$, $\forall c^*$; Exchange $\mathbf{X}_{c^*}^{(t)}$ with the other cells;
			\STATE {\bf If} $\sum_{c^{\ast} =1 }^{C} \|\mathbf{X}_{c^*}^{(t)} - \mathbf{X}_{c^*}^{(t-1)} \|_F \leq \delta$, {\bf then} Go to step 5;
			\STATE {\bf else if} {$\sum_{c^{\ast} =1 }^{C}  \|\mathbf{X}_{c^*}^{(t)}- \mathbf{X}_{c^*}^{(t-1)}\|_F>\delta$},
			{\bf then} $t=t+1$; Go to step 2;
			\STATE {\bf Outputs:} $\mathbf{X}_{c^*}^{\star}\leftarrow \mathbf{X}_{c^*}^{(t)}$, $\forall c^*$.
			
		\end{algorithmic}
	\end{algorithm}
	
	As the main computational complexity of Algorithm~\ref{alg2} is to solve \eqref{prob_06} at each BS, we analyze such complexity in the sequel. Since \eqref{prob_06} contains LMI constraints, a standard interior-point method (IPM) \cite{Boyd_convex} can be used to find its optimal solution. Therefore, we consider the worst-case runtime of the IPM to analyze the computational complexities of the proposed problem \eqref{prob_06} as follows.
	
	{\it Definition 1:} For a given $\epsilon >0$, the set of $\mathbf{X}_{c^*}^{(t),\epsilon},
	\mathbf{G}_{c^*}^{(t),\epsilon},\mathbf{A}_{c^*}^{(t),\epsilon}$ is called an $\epsilon$-solution to \eqref{prob_06} if
	\begin{equation}
	\textrm{Tr}\left( \mathbf{G}_{c^*}^{(t),\epsilon}\right)\leq \displaystyle \underset{ \mathbf{X}_{c^*}^{(t)},
		\mathbf{G}_{c^*}^{(t)},\mathbf{A}_{c^*}^{(t)}}{\textrm{minimize}}
	\textrm{Tr}\left(\mathbf{G}_{c^*}^{(t)}\right)+\epsilon.
	\end{equation}
	It can be observed that the number of decision variables of problem \eqref{prob_06} is on the order of $(2N+\tau)N$. Let $m=\mathcal{O}\left( (2N+\tau)N\right)$, we introduce the following lemma.
	\begin{lemma}
		\label{comlemma}
		The computational complexities to obtain $\epsilon$-solution to problem \eqref{prob_06} is
		\begin{eqnarray}
		\label{complexAI}
		\ln(\epsilon^{-1})\sqrt{4N+\tau} \alpha m,
		\end{eqnarray}
		where $\alpha=10N^3+(3\tau+6m)N^2+Nm\tau(m\tau+2)+\tau^2(m+\tau)+m^2.$
	\end{lemma}
	\begin{proof}
		Problem \eqref{prob_06} has  $1$ LMI constraint of dimension $N+\tau$, $1$ LMI constraint of dimension $2N$, and $1$ LMI constraint of dimension $N$. Based on these observations, one can follow the same steps as in \cite[Section V-A]{Kun-Yu2014} to arrive at \eqref{complexAI}. Note that the term $\ln(\epsilon^{-1})\sqrt{4N+\tau}$ in \eqref{complexAI} is the iteration complexities \cite{Kun-Yu2014} required for obtaining $\epsilon$-solutions to problem \eqref{prob_06} while the remaining terms represent the per-iteration computation costs \cite{Kun-Yu2014}.
	\end{proof}
	
	\section{Simulation Results}
	
	A wrapped-around multi-cell Massive MIMO system is considered for simulations with $C=4$, $M=500,$ and $N=10$. All users are randomly distributed over the coverage area. However, the distance between any user $n$ of cell $c$ and BS $c^{\ast}$, denoted as $d_{c,c^\ast}^n$, is always satisfied $d_{c,c^\ast}^n \geq 0.035$ km. The system utilizes $20$ MHz bandwidth related to the noise variance of $-96$ dBm and the noise figure of $5$ dB. The large-scale fading coefficient $\phi_{c,c\ast}^n$ [dB] is modeled as	$\phi_{c,c^\ast}^n = -148.1 - 37.6 \log_{10} (d_{c,c^\ast}^n) + z_{c,c^\ast}^n,$ where $z_{c,c^\ast}^n$ is the shadow fading following a log-normal Gaussian distribution with the standard variation of $7$ dB. Monte-Carlo simulations are tackled over $200$ different realizations of user locations. The widely adopted orthogonal pilot design, e.g., \cite{You2015a,Xu2015a},  is used as a benchmark where each pilot symbol is allocated with power $200$ mW and those orthogonal pilots are reused amongst users in the network. For every realization of user locations, such pilot signals are generated by the eigenvectors of a uniformly generated random matrix. The power constraint for pilot signal is set to be $P_{\textrm{max},c}=200\tau$ mW, $\forall c$.
	\begin{figure}[t]
		\centering
		\includegraphics[width=.4\textwidth]{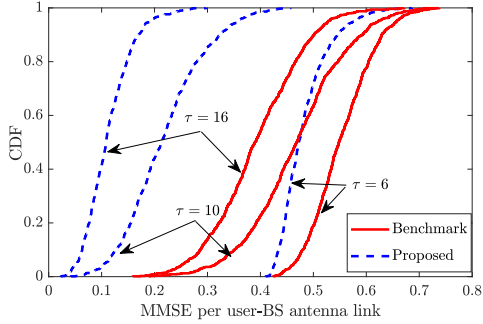} 
		\caption{The CDFs of the MMSE per user-BS antenna link for the proposed and benchmark approaches.}
		\label{Fig-CDFMMSE}
		\vspace*{-0.3cm}
	\end{figure}
	\begin{figure}[t]
		\centering
		\includegraphics[width=.4\textwidth]{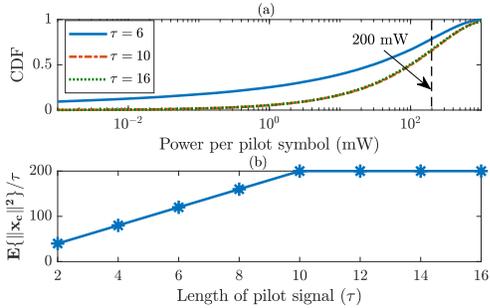} 
		\caption{$(a)$ The CDFs of the power allocated for each pilot symbol of the proposed approach with different pilot lengths; $(b)$ The average power allocated for each pilot symbol of the proposed approach v.s. the pilot length.}
		\label{Fig-Power}
		\vspace*{-0.3cm}
	\end{figure}
	\begin{figure}[t]
		\centering
		\includegraphics[width=.4\textwidth]{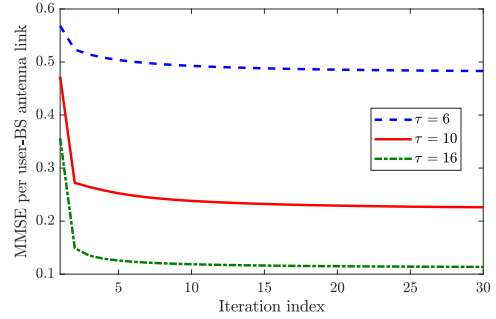} \vspace*{-0.2cm}
		\caption{The convergence of MMSE per user-BS antenna link v.s. iteration index for the proposed approach.}
		\label{Fig-Convergence}
		\vspace*{-0.3cm}
	\end{figure}
	
	Fig.~\ref{Fig-CDFMMSE} shows the cumulative distribution function (CDF) of the BS-user antenna link which is defined by $f_{c^\ast} (\mathbf{X}_{c^\ast})/(MN)$. It is clear from the figure that the channel estimation accuracy of the proposed approach is significantly improved compared to that of the benchmark. This confirms the effectiveness of our optimal pilot design in combating pilot contaminations. The results also indicate that the performance gap between the proposed approach and the benchmark increases as the pilot length increases. This is because increasing pilot length gives more degrees of freedom to the proposed approach for optimizing its performances.\footnote{The MMSE of the system approaches zero when the pilot length goes to infinity. An ideal pilot length of $\tau  = NC$ is sufficient to distinguish all users in the network, and also to balance between channel estimation errors and spectral efficiency (SE). However, this is impractical for a large-scale network. To that end, the proposed approach offers significant MMSE improvements for practical pilot lengths, i.e., when $\tau  < NC$.}
	
	Fig.~\ref{Fig-Power}~(a) displays the CDFs of the power allocated to each pilot symbol with different pilot lengths. It can be seen from the figure that for most of the cases, e.g., around $80 \ \%$ for $\tau=6$ and $70 \ \%$ for $\tau=10$ and $\tau=16$, the proposed approach spends less power for each symbol than the benchmark does, i.e., less than $200$ mW per symbol. Fig.~\ref{Fig-Power}~(b) illustrates the average power for each pilot symbol against the length of pilot signals. When $\tau < N=10$, the proposed approach consumes less power than the benchmark does and the system suffers from both intra-cell and inter-cell interference. In such hostile situations, the proposed approach can still effectively handle the pilot contaminations which results in higher channel estimate accuracy as seen in Fig.~\ref{Fig-CDFMMSE}, while consuming less power than the benchmark. Interestingly, when $\tau \geq N$, the optimal transmit power spent on each symbol turns out to be $200$ mW and to be constant irrespective to the value of $\tau$, which can also be observed from the CDF shown in Fig.~\ref{Fig-Power}~(a).

Fig.~\ref{Fig-Convergence} numerically reveals the fast convergent speed of the proposed approach, i.e., within less than 20 iterations. This results confirm the statement in Remark 1. Finally, Table~\ref{Table1} demonstrates the average uplink SE of per user in the network using the use-and-then-forget capacity bounding technique \cite[eq.~(28)]{Chien2018a} with maximum ratio detection, fixed data power $200$ mW, and the coherence block length $200$ symbols. Thanks to minimizing MSE, the proposed approach attains a higher uplink SE than the benchmark does.
	\begin{table}
		\caption{Average uplink SE of each user in the network [bits/s/Hz] for different pilot lengths.}
		\label{Table1}
		\centering
		\begin{tabular}{|c|c|c|c|}
			\hline
			& $\tau =6$ & $\tau = 10$& $\tau = 16$\\
			\hline
			Benchmark & 1.19 & 1.49 & 1.55 \\
			\hline
			Proposed & 1.39 & 1.74 & 1.86\\
			\hline
			Gain & 16.81\% & 16.78\% & 20.00\%\\
			\hline
		\end{tabular}
		\vspace{-0.5cm}
	\end{table}


\begin{thebibliography}{10}
		\providecommand{\url}[1]{#1}
		\csname url@samestyle\endcsname
		\providecommand{\newblock}{\relax}
		\providecommand{\bibinfo}[2]{#2}
		\providecommand{\BIBentrySTDinterwordspacing}{\spaceskip=0pt\relax}
		\providecommand{\BIBentryALTinterwordstretchfactor}{4}
		\providecommand{\BIBentryALTinterwordspacing}{\spaceskip=\fontdimen2\font plus
			\BIBentryALTinterwordstretchfactor\fontdimen3\font minus
			\fontdimen4\font\relax}
		\providecommand{\BIBforeignlanguage}[2]{{%
				\expandafter\ifx\csname l@#1\endcsname\relax
				\typeout{** WARNING: IEEEtran.bst: No hyphenation pattern has been}%
				\typeout{** loaded for the language `#1'. Using the pattern for}%
				\typeout{** the default language instead.}%
				\else
				\language=\csname l@#1\endcsname
				\fi
				#2}}
		\providecommand{\BIBdecl}{\relax}
		\BIBdecl
		
		
		\bibitem{Bjornson2016b}
		E.~Bj\"{o}rnson \emph{et al.}, ``Massive MIMO: 10 Myths and One Critical Question",vol.~54, no.~2, pp. 114-123, 2016.
		
				
		\bibitem{Zhang2}
		H.~Zhang \emph{et al.}, ``On superimposed pilot for
		channel estimation in multicell multiuser mimo uplink: Large system
		analysis,'' \emph{IEEE Trans. Vehi. Techno.}, vol.~65, no.~3, pp. 1492--1505, 2016.
		\bibitem{Ma}
		J.~Ma \emph{et al.}, ``On Orthogonal and Superimposed Pilot Schemes in Massive MIMO NOMA Systems," \emph{IEEE Journal on Selected Areas in Commun.}, vol.~35, no.~12, pp. 2696-2707, 2017.
		
		
		\bibitem{Xu2015a}
		X.~Zhu \emph{et al.}, ``Smart pilot assignment for {M}assive
		{MIMO},'' \emph{IEEE Commun. Letters}, vol.~19, no.~9, pp. 1644 - 1647,
		2015.
		
		\bibitem{PWang}
		P.~Wang \emph{et al.}, ``A Novel Pilot Assignment Approach for Pilot Decontaminating in Massive MIMO Systems, '' in \emph{Proc.} IEEE Wireless Commun. and Networking Conf. (WCNC), pp. 1-6, 2017.
		
		\bibitem{Chien2018a}
		T.~V.~Chien \emph{et al.}, ``Joint Pilot Design and Uplink Power Allocation in Multi-Cell Massive MIMO Systems,'' \emph{IEEE Trans. Wireless Commun.}, early access.
		
		\bibitem{Mav1}
		S.~Ma, \emph{et al.}, ``A Novel Pilot Assignment Scheme in Massive MIMO Networks," \emph{IEEE Wireless Commun. Letters}, early access.
		
		\bibitem{Boyd}
		M.~C. Grant and S.~P. Boyd, \emph{The CVX Users' Guide, Release 2.1.}, Mar.
		2015, [Online]. Available: \url{http://web.cvxr.com/cvx/doc/CVX.pdf}.
		
		\bibitem{Kailath}
		T.~Kailath \emph{et al.}, \emph{Linear Estimation}.\hskip 1em
		plus 0.5em minus 0.4em\relax Prentice Hall, 2000.
		
		\bibitem{Boyd_convex}
		S.~Boyd and L.~Vandenberghe, \emph{Convex Optimization}.\hskip 1em plus 0.5em
		minus 0.4em\relax Cambridge University Press, 2004.
		
		
		\bibitem{Kun-Yu2014}
		K.-Y. Wang \emph{et al.}, ``Outage
		constrained robust transmit optimization for multiuser {MISO} downlinks:
		Tractable approximations by conic optimization,'' \emph{IEEE Trans. Signal
			Process.}, vol.~62, no.~21, pp. 5690-5705, 2014.
		
		\bibitem{You2015a}
		L.~You \emph{et al.}, ``Pilot reuse for {M}assive
		{MIMO} transmission over spatially correlated rayleigh fading channels,''
		\emph{IEEE Trans. Wireless Commun.}, vol.~14, no.~6, pp. 3352-3366,
		2015.
		
	\end{thebibliography}
\end{document}